\documentclass[sigconf]{acmart}

\usepackage{graphicx}
\usepackage{amsmath,amsfonts}

\usepackage{amssymb}
\usepackage{mathtools}
\usepackage{bm}
\usepackage{subcaption}
\usepackage{booktabs}
\usepackage{multirow}
\usepackage{algorithm}
\usepackage{algpseudocode} 
\usepackage{tikz}
\usepackage{pgfplots}
\usepackage{xcolor}
\usepackage{placeins}
\usepackage{float}
\usepackage{enumerate}

\theoremstyle{plain}
\newtheorem{theorem}{Theorem}[section]
\newtheorem{proposition}[theorem]{Proposition}
\newtheorem{corollary}[theorem]{Corollary}
\theoremstyle{definition}
\newtheorem{definition}[theorem]{Definition}

\theoremstyle{remark}
\newtheorem{remark}[theorem]{Remark}
\newtheorem{lemma}[theorem]{Lemma}

\copyrightyear{2026}
\acmYear{2026}
\setcopyright{cc}
\setcctype{by-nc-nd}
\acmConference[HSCC '26]{29th ACM International Conference on Hybrid Systems: Computation and Control}{May 11--14, 2026}{Saint Malo, France}
\acmBooktitle{29th ACM International Conference on Hybrid Systems: Computation and Control (HSCC '26), May 11--14, 2026, Saint Malo, France}
\acmDOI{10.1145/3801146.3805668}
\acmISBN{979-8-4007-2566-1/2026/05}

\begin{document}

\title{Data-Driven Reachability Analysis Using \\Matrix Perturbation Theory}

\author{Peng Xie}
\authornote{These authors contributed equally to this work.}
\affiliation{%
  \department{Department of Computer Engineering}
  \institution{TUM School of Computation, Information and Technology, Technical University of Munich}
  \city{Heilbronn}
  \country{Germany}}
\email{p.xie@tum.de}

\author{Abdulla Fawzy}
\authornotemark[1] 
\affiliation{%
  \department{Department of Computer Engineering}
  \institution{TUM School of Computation, Information and Technology, Technical University of Munich}
  \city{Heilbronn}
  \country{Germany}}
\email{abdo1220032@gmail.com}

\author{Zhen Zhang}
\affiliation{%
  \department{Department of Computer Engineering}
  \institution{TUM School of Computation, Information and Technology, Technical University of Munich}
  \city{Heilbronn}
  \country{Germany}}
\email{zhenzhang.zhang@tum.de}

\author{Amr Alanwar}
\affiliation{%
  \department{Department of Computer Engineering}
  \institution{TUM School of Computation, Information and Technology, Technical University of Munich}
  \city{Heilbronn}
  \country{Germany}}
\email{alanwar@tum.de}

\begin{abstract}
We propose a matrix zonotope perturbation framework that leverages matrix perturbation theory to characterize how noise-induced distortions alter the dynamics within sets of models. The framework derives interpretable Cai--Zhang bounds for matrix zonotopes (MZs) and extends them to constrained matrix zonotopes (CMZs). Motivated by this analysis and the computational burden of CMZ-based reachable-set propagation, we introduce a coefficient-space approximation in which the constrained coefficient space of the CMZ is over-approximated by an unconstrained zonotope. Replacing CMZ--constrained-zonotope (CZ) products with unconstrained MZ--zonotope multiplication yields a simpler and more scalable reachable-set update. Experimental results demonstrate that the proposed method is substantially faster than the standard CMZ approach while producing reachable sets that are less conservative than those obtained with existing MZ-based methods, advancing practical, accurate, and real-time data-driven reachability analysis.
\end{abstract}

\maketitle

\section{Introduction}

Reachability analysis is a cornerstone for guaranteeing safety and performance in control systems. However, classical approaches typically assume an accurate model, which is often unavailable or costly to identify. Recent work has therefore turned to \emph{data-driven} reachability analysis, where instead of identifying a single model, a set of models is propagated directly from data. Matrix zonotopes (MZs) have emerged as a convenient representation in this context, enabling algorithms that compute over-approximations of reachable sets for unknown linear (and even Lipschitz) systems directly from noisy data with set-containment guarantees~\cite{alanwar2023data,alanwar2023robust}.

A key appeal of MZs is that they compactly encode uncertainty in linear operators while preserving desirable closure properties under linear algebraic operations. Recently, this approach has been extended to hybrid systems~\cite{xie2025data} using hybrid zonotopes~\cite{bird2023hybrid}, nonlinear systems~\cite{zhang2025data} with constrained polynomial zonotopes~\cite{kochdumper2023constrained}, time-varying systems~\cite{akhormeh2025online}, line zonotopes~\cite{rego2025line}, and the corresponding control methods~\cite{oumer2025data, farjadnia2024robust, russo2022tube}. Constrained zonotopes (CZs)~\cite{scott2016constrained} further reduce conservatism by imposing linear consistency constraints on the coefficient space, and have demonstrated favorable accuracy--efficiency trade-offs in estimation and set-propagation tasks~\cite{scott2016constrained}. These ideas have been instantiated in data-driven reachability via \emph{constrained matrix zonotopes} (CMZs), which empirically deliver tighter sets but at a significant computational cost~\cite{alanwar2023data}. In particular, the number of generators in the reachable sets can grow rapidly across time steps, exacerbating both memory usage and runtime; this motivates principled order-reduction and approximation strategies~\cite{kopetzki2017methods}.

Despite this progress, a systematic understanding of \emph{how process noise perturbs the matrices} encoded by an MZ or CMZ, and consequently how such perturbations shape the propagation of reachable sets, remains lacking. This paper closes that gap by importing tools from \emph{matrix perturbation theory}. At a high level, perturbations rotate the subspaces of the system matrix; these geometric changes directly modulate the directions and magnitudes of the matrix-zonotope generators that drive reachability growth. We leverage classical and sharp bounds, namely
Wedin perturbation bounds for singular subspaces~\cite{wedin1972perturbation} and recent Cai--Zhang bounds that are tight up to constants~\cite{cai2018rate}---to obtain interpretable stability estimates for matrix-zonotope-based reachable-set propagation under noise.

\begin{figure*}[!t]
\centering
\includegraphics[width=0.95\textwidth]{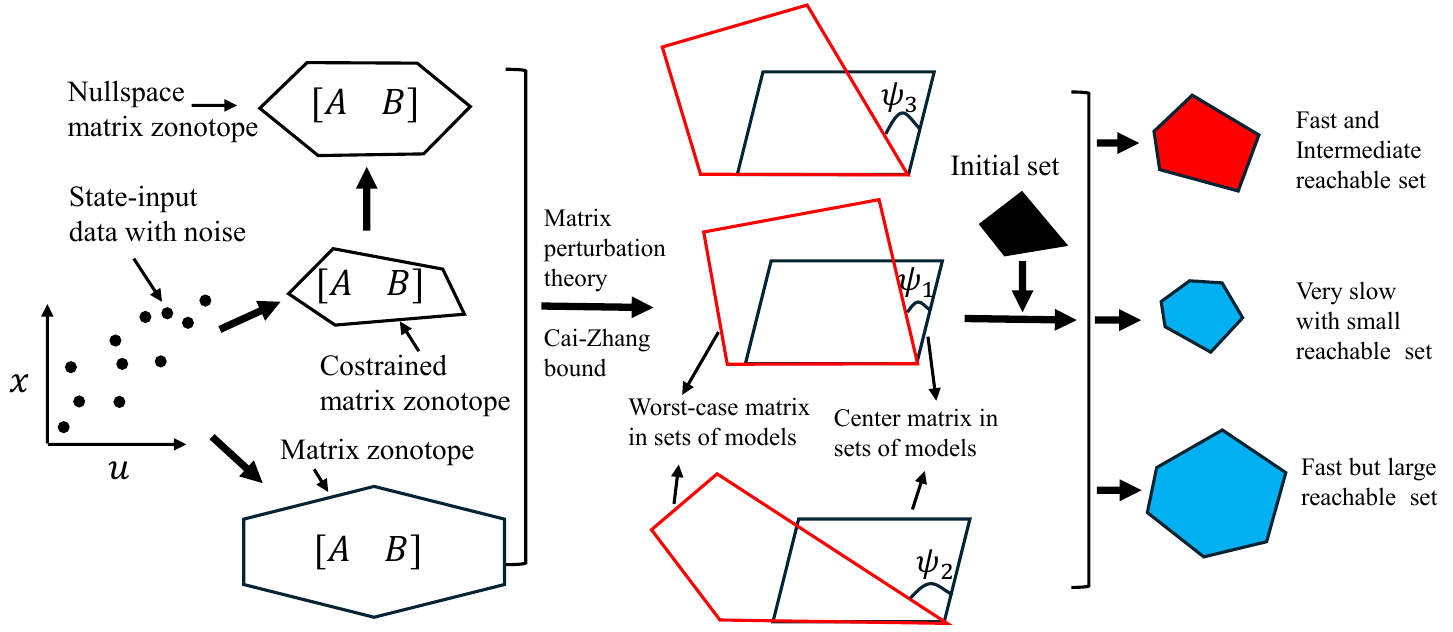}
\caption{Noisy state--input data give rise to a set of models for the pair $[A\ B]$. 
  Matrix perturbation theory, and in particular the Cai--Zhang bound, are used to quantify the worst-case rotation of the relevant subspaces across this set of models. 
  A CMZ keeps the subspace rotation small and yields a tight (small) reachable set, but is computationally expensive to propagate. 
  A MZ has a much simpler structure and is cheap to update, but allows large subspace rotations and therefore leads to a conservative (large) reachable set. 
  The proposed nullspace matrix zonotope (NMZ) aims to strike a balance between these extremes: 
  it retains an MZ-like structure while controlling the subspace rotation so that the resulting reachable set remains sufficiently tight.}
\label{fig:main_workflow}
\end{figure*}

As illustrated in Figure~\ref{fig:main_workflow}, our data-driven setting does not yield a single deterministic model but rather an entire set of models.
The figure highlights that the size of the resulting reachable set is largely determined by how much the relevant subspaces are allowed to rotate across this model set.
A CMZ restricts these rotations and therefore yields small reachable sets, but propagation through it is computationally expensive.
Conversely, propagation using an MZ is significantly faster yet permits large subspace rotations, leading to a conservative reachable set.
This motivates our intermediate representation, the nullspace matrix zonotope (NMZ), which retains an MZ-like structure while explicitly controlling the maximal subspace rotation.
The Cai--Zhang singular-subspace perturbation bound is the tool that makes this possible: it provides a computable worst-case rotation bound for the model set and guides the construction of the NMZ.

Our contributions are as follows. (1) We introduce a \emph{matrix zonotope perturbation framework} for data-driven reachability analysis that characterizes, starting from noisy trajectories, how perturbations across the set of models translate into directional changes and size inflation.
(2) We extend the perturbation analysis to \emph{constrained matrix zonotopes} (CMZs), clarifying how constraints materially reduce conservatism. (3) We propose the \emph{nullspace matrix zonotope} (NMZ), a CMZ-to-MZ coefficient-space approximation that mitigates generator growth across propagation steps by over-approximating the feasible coefficient set of a CMZ with an unconstrained zonotope and using it to construct a new MZ.

The remainder of the paper is organized as follows. Section~\ref{sec-prelim} reviews notation and preliminaries on zonotopes, CZs, MZs, CMZs, and matrix perturbation theory.
Section~\ref{sec:cz-matzono-score} develops Cai--Zhang-based bounds for evaluating matrix zonotopes and extends the framework to CMZs.
Section~\ref{sec:nmz} introduces the nullspace matrix zonotope for over-approximating the CMZ. Section~\ref{sec:exp} presents experimental results, and Section~\ref{sec:conclus} concludes with future directions.

\section{Notations and Problem Statement}\label{sec-prelim}
This section establishes the notation, set representations, and spectral tools used throughout the paper.

\subsection{Notation}

Matrices are uppercase Roman ($A$), sets are calligraphic ($\mathcal{R}_k$), and vectors/scalars are lowercase ($x_k$, $\alpha_i$). We write $\mathbf{1}_n$, $\mathbf{0}_n$ for the all-ones/zeros vectors, $I_n$ for the identity, and $\mathbf{1}_{n\times m}$, $\mathbf{0}_{n\times m}$ for all-ones/zeros matrices. For $A\!\in\!\mathbb{R}^{m\times n}$: $A_{(i,j)}$ is the $(i,j)$-th entry, $A_{(:,j)}$ the $j$-th column, $A_{(i,:)}$ the $i$-th row, and $A_{(i:j,\,:)}$ the submatrix of rows $i$ through $j$. We denote the transpose, Moore--Penrose pseudoinverse, and nullspace basis of $X$ by $X^\top$, $X^\dagger$, and $X^\perp$, respectively, and write $\otimes$ for the Kronecker product, $\mathrm{vec}(X)$ for vectorization, and $\mathrm{diag}(a)$ for the diagonal matrix formed from $a$. The nullspace of $X$ is $\mathcal{N}(X)$ with dimension $\mathrm{nullity}(X)$; $|X|$ denotes the element-wise absolute value. Superscripts index elements: $a^{(i)}$ is the $i$-th entry, $a^{(i:j)}$ a range.

The spectral and Frobenius norms are $\|M\|_2$ and $\|M\|_F$. Singular values are $\sigma_1(M)\!\ge\!\cdots\!\ge\!\sigma_{\min(m,n)}(M)$; eigenvalues of symmetric $A$ are $\lambda_1(A)\!\ge\!\cdots\!\ge\!\lambda_n(A)$. The notation $\sin\Theta(U_1,U_2)$ is the matrix of canonical angles between the column spaces of $U_1$ and $U_2$, and $a\wedge b=\min(a,b)$. The Minkowski sum is $\mathcal{A}\oplus\mathcal{B}=\{a+b:a\!\in\!\mathcal{A},b\!\in\!\mathcal{B}\}$, and the Cartesian product is $\mathcal{A}\times\mathcal{B}$. For $X\!\in\!\mathbb{R}^{n\times p}$, its orthogonal complement $X_\perp\!\in\!\mathbb{R}^{n\times(n-p)}$ satisfies $[X\ X_\perp]\!\in\!\mathcal{O}(n)$.

\subsection{Set representations}
Data-driven reachability analysis relies on several set representations that encode both state sets and model uncertainty. We recall their definitions below.
A \emph{zonotope}~\cite{girard2005reachability} with center $c_{\mathcal{Z}}\!\in\!\mathbb{R}^{n_x}$ and generator matrix $G_{\mathcal{Z}}\!\in\!\mathbb{R}^{n_x\times\gamma_{\mathcal{Z}}}$ is
\begin{equation}\label{def:zonotopes}
\mathcal{Z}=\langle c_{\mathcal{Z}},G_{\mathcal{Z}}\rangle=\bigl\{c_{\mathcal{Z}}+G_{\mathcal{Z}}\xi:\|\xi\|_\infty\le 1\bigr\}.
\end{equation}
A \emph{constrained zonotope} (CZ)~\cite{scott2016constrained} adds a linear equality constraint on the coefficients:
\begin{equation}\label{def:constrained_zonotope}
\mathcal{Z}=\langle G,c,A,b\rangle=\bigl\{c+G\xi:\|\xi\|_\infty\le 1,\;A\xi=b\bigr\}.
\end{equation}
Standard operations (linear map, Minkowski sum, generalized intersection) on CZs are given in~\cite{scott2016constrained}.

A \emph{matrix zonotope} (MZ)~\cite{althoff2010reachability} extends the zonotope to matrix spaces. Given a center $C_{\mathcal{M}}\!\in\!\mathbb{R}^{n\times m}$ and $\gamma$ generators $G^{(i)}_{\mathcal{M}}\!\in\!\mathbb{R}^{n\times m}$:
\begin{equation}\label{eq:matrix_zonotope}
\mathcal{M}=\langle C_{\mathcal{M}},\{G^{(i)}_{\mathcal{M}}\}\rangle=\bigl\{C_{\mathcal{M}}+\textstyle\sum_{i=1}^{\gamma}\xi^{(i)}G^{(i)}_{\mathcal{M}}:\|\xi\|_\infty\le 1\bigr\}.
\end{equation}
A \emph{constrained matrix zonotope} (CMZ)~\cite{alanwar2023data} additionally imposes $\hat{A}_{\mathcal{N}}\xi=\hat{b}_{\mathcal{N}}$:
\begin{multline}\label{def:new-cmz}
\mathcal{N}=\langle C_{\mathcal{N}},\tilde{G}_{\mathcal{N}},\hat{A}_{\mathcal{N}},\hat{b}_{\mathcal{N}}\rangle=\bigl\{C_{\mathcal{N}}+\textstyle\sum_{i=1}^{\gamma_{\mathcal{N}}}\xi^{(i)}G_{\mathcal{N}}^{(i)}:\\
\hat{A}_{\mathcal{N}}\xi=\hat{b}_{\mathcal{N}},\;\|\xi\|_\infty\le 1\bigr\},
\end{multline}
where $\tilde{G}_{\mathcal{N}}=[G_{\mathcal{N}}^{(1)}\cdots G_{\mathcal{N}}^{(\gamma_{\mathcal{N}})}]\!\in\!\mathbb{R}^{n_x\times(p\gamma_{\mathcal{N}})}$.
Converting the original CMZ definition in~\cite{alanwar2023data} to the form above amounts to setting $\hat{A}_{\mathcal{N}}=[\mathrm{vec}(A^{(1)}_{\mathcal{N}})\;\cdots\;\mathrm{vec}(A^{(\gamma_{\mathcal{N}})}_{\mathcal{N}})]$ and $\hat{b}_{\mathcal{N}}=\mathrm{vec}(B_{\mathcal{N}})$.
The coefficient space of $\mathcal{N}$ is
\begin{equation}\label{eq:coefficent_space}
    \Xi = \{\xi \mid \hat{A}_{\mathcal{N}} \xi = \hat{b}_{\mathcal{N}}, \|\xi\|_\infty \leq 1\}.
\end{equation}
\begin{proposition}[Multiplication of CMZ by a matrix ~\cite{alanwar2023data}]
    For every $\mathcal{N}_1 {=} \left\langle{C_{\mathcal{N}_1},\{G_{\mathcal{N}_1}^{(i)}\}_{i=0}^{\gamma_{\mathcal{N}_1}},\hat{A}_{\mathcal{N}_1},B_{\mathcal{N}_1}} \right\rangle \subset \mathbb{R}^{n_x \times p}$, and $R \,{\in} \, \mathbb{R}^{k \times n_x}$ the following identities hold
    \begin{align}
        R\mathcal{N}_1 &= \left\langle {RC_{\mathcal{N}_1},\{RG_{\mathcal{N}_1}^{(i)}\}_{i=0}^{\gamma_{\mathcal{N}_1}},\hat{A}_{\mathcal{N}_1},B_{\mathcal{N}_1}} \right\rangle \label{eq:RN-cmz} \ .
    \end{align}
\end{proposition}

\begin{proposition}[CMZ-CZ Multiplication~\cite{alanwar2023data}]
\label{prop:cmz_cz}
Given a constrained matrix zonotope $\mathcal{N}=\langle C_\mathcal{N},\{G^{(1)}_\mathcal{N},\ldots,G^{(\gamma)}_\mathcal{N}\},\hat A_\mathcal{N},\hat b_\mathcal{N}\rangle$ with $C_\mathcal{N} \in \mathbb{R}^{n\times m}$, $G_\mathcal{N}^{(i)} \in \mathbb{R}^{n\times m}$, and a constrained zonotope $\mathcal{Z}=\langle G_z,c_z,A_z,b_z\rangle$ with $G_z=[g_z^{(1)}\ \cdots\ g_z^{(n_g)}]\in\mathbb{R}^{m\times n_g}$, $c_z\in\mathbb{R}^m$, the product admits a conservative CZ over-approximation:
\begin{equation}
\mathcal{N}\cdot\mathcal{Z}\subseteq\langle \bar G,\bar c,\bar A,\bar b\rangle,
\end{equation}
where the parameters are defined as:
\begin{align}
\bar c &= C_\mathcal{N} \,c_z, \quad \bar G = \big[\, G_\mathcal{N}^{(1)}c_z\ \cdots\ G_\mathcal{N}^{(\gamma)}c_z\  
C_\mathcal{N} G_z\  G_f \,\big],\nonumber\\
\bar A &= \mathrm{blkdiag}(\hat A_\mathcal{N},\ A_z),\quad
\bar b = \big[\hat b_\mathcal{N}^\top\ \ b_z^\top\big]^\top, \label{eq:cmz_cz_A}
\end{align}
with cross generators $G_f=\big[g_f^{(i,j)}\big]_{\substack{i=1,\ldots,\gamma\\j=1,\ldots,n_g}}$ scaled as:
\begin{align}
g_f^{(i,j)} &= \bar d^{(i,j)}\, G_\mathcal{N}^{(i)}\, g_z^{(j)}, \label{eq:cmz_cz_cross}\\
\bar d^{(i,j)} &= \max\big\{|\xi_{L,N}^{(i)}\xi_{L,z}^{(j)}|,\ 
|\xi_{U,N}^{(i)}\xi_{U,z}^{(j)}|, \nonumber\\
&\qquad\quad|\xi_{U,N}^{(i)}\xi_{L,z}^{(j)}|,\ 
|\xi_{L,N}^{(i)}\xi_{U,z}^{(j)}|\big\}, \nonumber
\end{align}
where $\xi_{L/U,N}^{(i)}$ and $\xi_{L/U,z}^{(j)}$ denote componentwise minima/maxima of the $i$-th and $j$-th generator coefficient of $\mathcal{N}$ and $\mathcal{Z}$ respectively.
\end{proposition}

\subsection{Spectral Perturbation Theory}
Quantifying how perturbations affect the subspaces of a matrix is central to our framework. We recall the Weyl inequality~\cite{weyl1912asymptotische}, which is used throughout this work.

\begin{theorem}[Weyl's Inequality with Hermitian Dilation~\cite{weyl1912asymptotische, tropp2015introduction}]\label{Wel_theory}
Let $A, E$ be $\mathbb{R}^{n \times n}$ matrices. Then for each $i = 1, \ldots, n$:
\begin{equation}
|\lambda_i(A + E) - \lambda_i(A)| \leq \|E\|_2.
\end{equation}
For singular values of general matrices $M, E \in \mathbb{R}^{m \times n}$:
\begin{equation}
|\sigma_i(M + E) - \sigma_i(M)| \leq \|E\|_2.
\end{equation}
\end{theorem}

\subsection{Problem Statement}
We now formalize the data-driven reachability problem.
Consider an LTI system
\begin{equation}
    \label{eq:LTI-sys}
    x(k+1) = A x(k) + Bu(k)+ w(k)
\end{equation}
where the matrices $A$ and $B$ are unknown. Instead of a known system model, we are given $K$ input-state trajectories, each of length $T_i+1$, denoted by $\{u^{(i)}(k)\}^{T_i-1}_{k=0}$ and $\{x^{(i)}(k)\}^{T_i}_{k=0}$ for $i=1,\dots,K$.
The following data matrices ~\cite{alanwar2023data} are defined:
\begin{align}\label{def:input-state_trajectories} X_+ =& \left[\!\!\begin{array}{cccccccccc}x^{(1)}(1)\dots x^{(1)}(T_{1}) \dots x^{(K)}(1) \dots x^{(K)}(T_{K}) \end{array}\!\!\right]\nonumber , \\ X_- =& \left[\!\!\begin{array}{cccccccccc}x^{(1)}(0) \dots x^{(1)}(T_{1}\!-\!1) \dots x^{(K)}(0) \dots x^{(K)}(T_{K}\!-\!1) \end{array}\!\!\right]\nonumber , \\ U_- =& \left[\!\!\begin{array}{cccccccccc}u^{(1)}(0) \dots u^{(1)}(T_{1}\!-\!1) \!\dots\! u^{(K)}(0) \dots u^{(K)}(T_{K}\!-\!1) \end{array}\!\!\right]\! \ . \end{align}
The total number of data points is denoted by $T = \sum_{i=1}^{K}{T_i}$.

The process noise is assumed to be bounded by a zonotope $\mathcal{Z}_w = \left\langle c_{\mathcal{Z}_w}, G_{\mathcal{Z}_w} \right\rangle = \left\langle c_{\mathcal{Z}_w}, \begin{bmatrix} g_{\mathcal{Z}_w}^{(1)} \cdots g_{\mathcal{Z}_w}^{(\gamma_{\mathcal{Z}_w})} \end{bmatrix} \right\rangle$ where $\gamma_{\mathcal{Z}_w}$ is the number of generators of $\mathcal{Z}_w$, i.e. $w(k) \in \mathcal{Z}_w, \ \forall k$. The input is bounded by an input zonotope $\mathcal{U}_k$, i.e. $u(k) \in \mathcal{U}_k, \ \forall k$. The initial state of the system is bounded by the initial set $\mathcal{X}_0$, i.e. $x(0) \in \mathcal{X}_0$.

The exact reachable set $\mathcal{R}_{N}$ after $N$ time steps, subject to inputs ${u(k) \in \mathcal{U}_k}$, $\forall k {=}\{ 0, \dots, N-1\}$, and noise $w( \cdot) \in \mathcal{Z}_w$, is the set of all state trajectories starting from the initial set $\mathcal{X}_0$ after $N$ steps:
\begin{align} \label{eq:R}
    \mathcal{R}_{N} = \big\{& x(N) \in \mathbb{R}^{n_x} \, \big| \ x(k{+}1) = f(x(k),u(k)) + w(k), \nonumber\\
   & \, x(0) \in \mathcal{X}_0,
    u(k) \in \mathcal{U}_k, w(k) \in \mathcal{Z}_w: \nonumber\\ & \forall k \in \{0,...,N{-}1\}\big\}.
\end{align}

Given a matrix zonotope $\mathcal{M}_\Sigma$ that contains all possible system models $[A \ \ B]$ consistent with~\eqref{eq:LTI-sys}, an over-approximation of the reachable sets can be computed via the relation:
\begin{equation}
    \label{eq:reach-propagation}
    \hat{\mathcal{R}}_{k+1} = \mathcal{M}_{\Sigma} (\hat{\mathcal{R}}_{k} \times \mathcal{U}_{k}  ) \oplus \mathcal{Z}_w
\end{equation}
where $\mathcal{R}_k \subseteq \hat{\mathcal{R}}_k$, provided that the true system model $[A_{\mathrm{tr}} \ \ B_{\mathrm{tr}}] \in \mathcal{M}_\Sigma$~\cite{alanwar2023data}.

As derived in ~\cite{alanwar2023data}, the tightest data-consistent model set is a constrained matrix zonotope $\mathcal{N}_\Sigma$ (as opposed to the generic matrix zonotope $\mathcal{M}_\Sigma$ used in~\eqref{eq:reach-propagation}), defined as
\begin{equation}
    \label{eq:N-sigma}
    \mathcal{N}_\Sigma = (X_+ - \mathcal{N}_w) \begin{bmatrix} X_- \\ U_- \end{bmatrix}^\dagger
\end{equation}
where $\mathcal{N}_w = \left\langle C_{\mathcal{N}_w}, \tilde{G}_{\mathcal{N}_w}, \hat{A}_{\mathcal{N}_w}, \hat{b}_{\mathcal{N}_w} \right\rangle$ where
\begin{align} C_{\mathcal {N}_{w}} =& \left[\!\!\begin{array}{cccccccccc}c_{\mathcal {Z}_{w}} & \dots & c_{\mathcal {Z}_{w}}\end{array}\!\!\right] , \\ \tilde{G}_{\mathcal {N}_{w}}=&\left[\!\!\begin{array}{cccccccccc}G_{\mathcal {N}_{w}}^{(1)}&\dots &G_{\mathcal {N}_{w}}^{(\gamma _{\mathcal {N}_{w}})}\end{array}\!\!\right] , \\ G^{(j+(i-1)T)}_{\mathcal {N}_{w}} =& \left[\!\!\begin{array}{cccccccccc}0_{n \times (j-1)} &g_{\mathcal {Z}_{w}}^{(i)} & 0_{n \times (T-j)}\end{array}\!\!\right] \label{def:noise-matrix-zono} \ . \end{align}
$\forall i \in \{ 1, \dots, \gamma_{\mathcal{Z}_w} \}$ and $\forall j \in \{ 1, \dots, T \}$ and $\gamma_{\mathcal{N}_w}$ is the number of generators of $\mathcal{N}_w$.
The constraint equation of $\mathcal{N}_w$, as derived in~\cite{alanwar2023data}, is
\begin{equation}
    \label{eq:const-eq}
    (X_+ - C_{\mathcal{N}_w}) \begin{bmatrix} X_- \\ U_- \end{bmatrix}^\perp = \sum_{i = 1}^{\gamma_{\mathcal{Z}_w}T}{\xi^{(i)} G^{(i)}_{\mathcal{N}_w} \begin{bmatrix} X_- \\ U_- \end{bmatrix}^\perp} .
\end{equation}

We therefore seek to quantify how $\mathcal{Z}_w$ shapes $\mathcal{N}_\Sigma$ in \eqref{eq:N-sigma}---viewed as a center plus a perturbation that rotates the relevant singular subspaces---and how this impacts the reachable-set propagation \eqref{eq:reach-propagation}.

\section{Cai–Zhang Bound-Based Evaluation of Matrix Zonotopes}
\label{sec:cz-matzono-score}
The quality of a matrix representation fundamentally depends on its ability to maintain stable subspace structure under perturbations~\cite{o2023matrices}. While the classical Wedin perturbation theorem~\cite{wedin1972perturbation} and the Davis--Kahan $\sin\Theta$ theorem~\cite{davis1970rotation} provide bounds on singular subspace variations, the Cai--Zhang bound offers significantly tighter estimates by exploiting the block structure of the perturbation matrix. This section develops a quantitative framework for evaluating the quality of matrix zonotopes based on these refined bounds.
\begin{theorem}[Cai-Zhang Bounds Theory~\cite{cai2018rate}]
\label{thm:cai-zhang}
Let $C,\hat C,Z \in \mathbb{R}^{m \times n}$, where $C=U\Sigma V^{\top}$, $\hat C=\hat U\,\hat \Sigma\,\hat V^\top$, and $Z=\hat C-C$. Define
\[
\alpha=\sigma_{\min}\!\bigl(U^{\top}\hat C V\bigr),\quad
\beta=\bigl\|U_{\perp}^{\top}\hat C V_{\perp}\bigr\|_2,
\]
\[
Z_{12}=U^{\top} Z V_{\perp},
Z_{21}=U_{\perp}^{\top} Z V,
z_{12}=\|Z_{12}\|_2,
z_{21}=\|Z_{21}\|_2.
\]
If $\alpha^2>\beta^2+\min\{z_{12}^2,z_{21}^2\}$, then
\begin{align}
\|\sin\Theta(V,\hat V)\|
&\le
\left(
\frac{\alpha\,z_{12}+\beta\,z_{21}}
     {\alpha^2-\beta^2-\min\{z_{12}^2,z_{21}^2\}}
\right)\wedge 1,
\label{eq:cz-right}\\[-1mm]
\|\sin\Theta(U,\hat U)\|
&\le
\left(
\frac{\alpha\,z_{21}+\beta\,z_{12}}
     {\alpha^2-\beta^2-\min\{z_{12}^2,z_{21}^2\}}
\right)\wedge 1.
\label{eq:cz-left}
\end{align}
\end{theorem}

\begin{remark}[Full-row-rank] If $\operatorname{rank}(C)=m$, then $U_{\perp}=\varnothing$, hence
$\beta=0$ and $z_{21}=0$, which implies $\|\sin\Theta(U,\hat U)\|=0$ and
\begin{equation}
\label{eq:tight-right}
\|\sin\Theta(V,\hat V)\|
~\le~
\frac{z_{12}}{\alpha}
~=~
\frac{\|U^{\top} Z V_{\perp}\|_2}
     {\sigma_{\min}\!\bigl(U^{\top}\hat C V\bigr)}.
\end{equation}
\end{remark}

In data-driven reachability analysis, one does not observe a single system matrix but rather an entire family of data-consistent models.
Each candidate can be written as $\hat{C} = C + E(\xi)$, where $E(\xi)$ is induced by bounded noise and $\xi$ lies in a (constrained) zonotope.
The Cai--Zhang bound links the perturbation $E$ to an explicit upper bound on the canonical angles between the subspaces of $C$ and $C + E$, thereby quantifying how much the relevant subspaces can rotate across the model set.
In the following, we express $E(\xi)$ in terms of the matrix generators and maximize the bound over the coefficient constraints to obtain a computable worst-case subspace perturbation for the entire zonotope.

\subsection{Explicit Bounds for Matrix Zonotopes}
We now specialize the Cai--Zhang bound to the matrix zonotope setting and derive a closed-form worst-case expression.

Consider a matrix zonotope with perturbation $E(\boldsymbol{\xi}) = \sum_{i=1}^{p} \xi^{(i)} G^{(i)}$, where $\boldsymbol{\xi} \in [-1,1]^p$. According to Theorem~\ref{thm:cai-zhang},
\[
Z_{12}(\boldsymbol{\xi}) = \sum_{i=1}^{p} \xi^{(i)} U^{\top}G^{(i)}V_{\perp}.\]
We derive explicit bounds through two steps. By linearity and the triangle inequality~\cite{maligranda2006simple},
\begin{align}
\label{eq:Z12xi}
 \|Z_{12}(\boldsymbol{\xi})\|_2 &\leq \sum_{i=1}^{p} |\xi^{(i)}|\mu^{(i)}, \quad \mu^{(i)} = \|U^{\top}G^{(i)}V_{\perp}\|_2.
\end{align}
The Weyl inequality (Theorem~\ref{Wel_theory}) provides $\sigma_{\min}(C + E) \geq \sigma_{\min}(C) - \|E\|_2$. Since $\|E(\boldsymbol{\xi})\|_2 \leq \sum_i |\xi^{(i)}|\gamma^{(i)}$ with $\gamma^{(i)} = \|G^{(i)}\|_2$:
\begin{equation}
\label{eq:gxi}
g(\boldsymbol{\xi}) = \sigma_{\min}(C) - \sum_{i=1}^{p} |\xi^{(i)}|\gamma^{(i)}.
\end{equation}
Combining these results yields the explicit bound function:
\begin{equation}
\label{eq:fxi}
f(\boldsymbol{\xi}) = \frac{\sum_i |\xi^{(i)}|\mu^{(i)}}{g(\boldsymbol{\xi})}.
\end{equation}

Since \eqref{eq:fxi} increases monotonically in each $|\xi^{(i)}|$, the maximum over $[-1,1]^p$ is attained at a vertex where $\xi^{(i)} = \pm 1$, yielding the worst-case bound:
\begin{equation}
\label{eq:vertex}
\max_{\boldsymbol{\xi}} \|\sin\Theta\| \leq \frac{\sum_{i=1}^{p} \mu^{(i)}}{\sigma_{\min}(C) - \sum_{i=1}^{p} \gamma^{(i)}},
\end{equation}

\begin{remark}[Data-Driven Matrix Zonotope Structure]
\label{rem:common-factor}
In the data-driven framework \cite{alanwar2023data}, matrix zonotope generators share a common structure:
\begin{equation}
\label{eq:zonotope-structure}
\mathcal{M}_{\Sigma} = C + \sum_{i=1}^{p} \xi^{(i)} G^{(i)}, \quad G^{(i)} = G_{\mathcal{M}_w}^{(i)}H,
\end{equation}
where $H = [X_-; U_-]^{\dagger}$ is the data pseudoinverse~\cite{laub-mpp}. This shared right factor enables the factorization:
\begin{equation}
\label{eq:Z12-factor}
Z_{12}(\xi)=\left(\sum_{i=1}^p \xi^{(i)} U^{\top} G_{\mathcal{M}_w}^{(i)}\right)\left(H V_{\perp}\right)
\end{equation}
\end{remark}

This factorization separates data-dependent terms from noise-dependent terms, yielding the global bound:
\begin{align}
\label{eq:global}
\|\sin\Theta\| &\leq \frac{\kappa \sum_i \mu^{(i)}_w}{\sigma_{\min}(C) - \sum_i \gamma^{(i)}},\\
\text{where } \kappa &= \|HV_{\perp}\|_2, \quad \mu^{(i)}_w = \|U^{\top}G_{\mathcal{M}_w}^{(i)}\|_2
\end{align}

\subsection{Impact of Data Scaling}
The explicit bounds derived above depend on the data matrices through the pseudoinverse $H$. We now examine how the magnitude of the state-input data influences the bound in~\eqref{eq:global}. Consider scaled data $\Phi_{\text{new}} = D\Phi$ where $\Phi = [X_-; U_-]$ and $D = \mathrm{diag}(d_1, \ldots, d_{n+m})$ with $d_{\min} = \min_i d_i > 1$. Using the pseudoinverse property $\|A^{\dagger}\|_2 = 1/\sigma_{\min}(A)$ \cite{matrix-cookbook,laub-mpp}, we can obtain a bound on the norm of the new $H$:
\begin{equation}
\|H_{\text{new}}\|_2 = \frac{1}{\sigma_{\min}(D\Phi)} \leq \frac{1}{d_{\min}\sigma_{\min}(\Phi)} = \frac{\|H\|_2}{d_{\min}}.
\end{equation}

Both $\gamma^{(i)} = \|G_{\mathcal{M}_w}^{(i)}H\|_2$ and the effective block norms scale with $\|H\|_2$, while $\sigma_{\min}(C)$ remains largely unaffected by data rescaling. As shown in Figures~\ref{fig:cz-models} and~\ref{fig:cz-scale}, the bound in~\eqref{eq:global} decreases monotonically as data magnitudes increase, and the corresponding reachable sets become less conservative.

\begin{figure}[t]
  \centering
  \includegraphics[width=\linewidth,page=1]{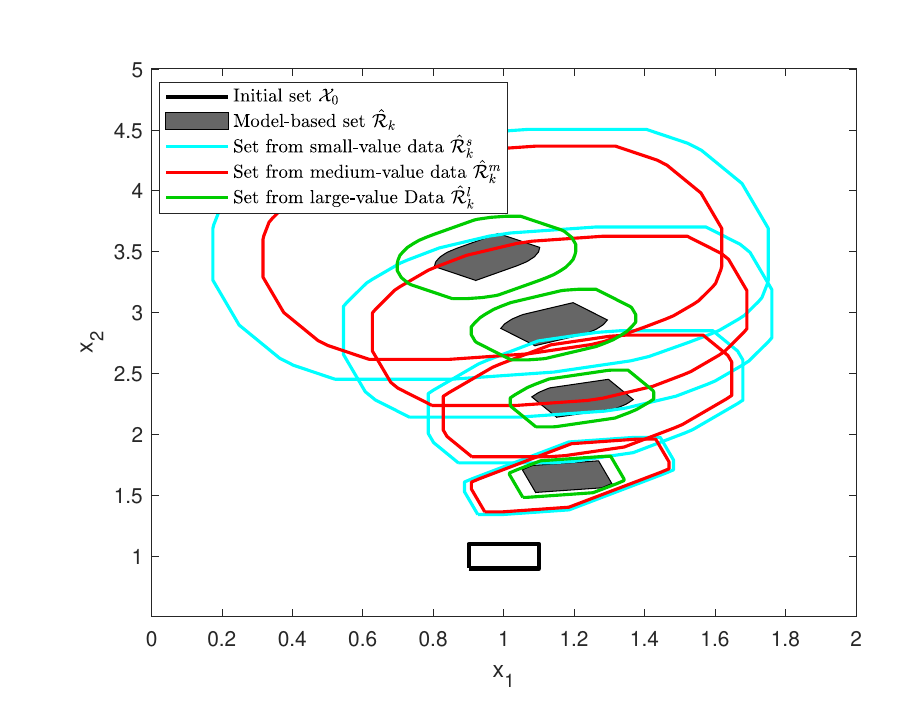}
  \caption{Reachable-set projections under different state-input data scalings: larger magnitudes reduce the upper bound on $\|\sin\Theta(V,\widehat V)\|$.}
  \label{fig:cz-models}
\end{figure}

\begin{figure}[t]
  \centering
  \includegraphics[width=0.9\linewidth,page=1]{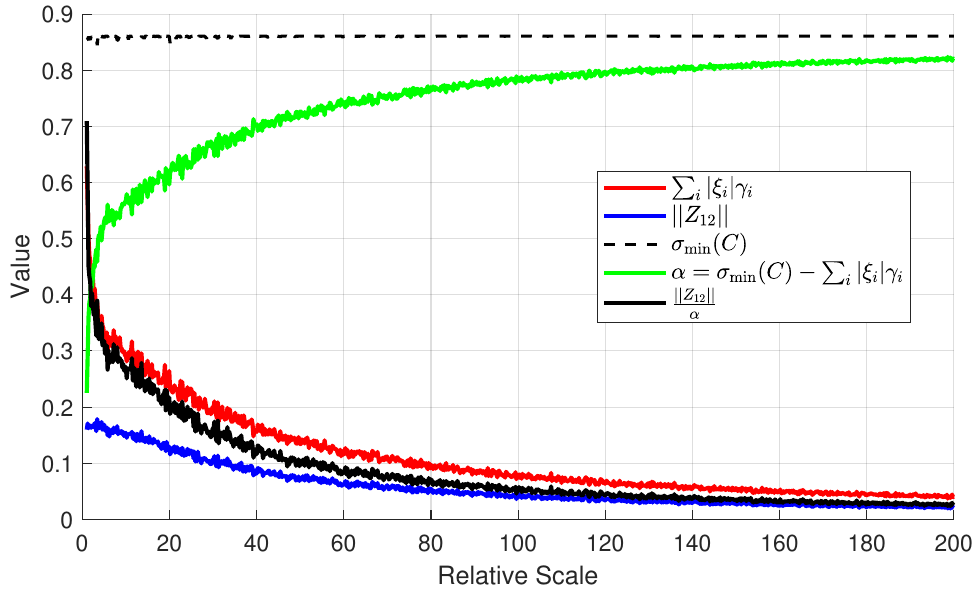}
  \caption{With fixed noise, enlarging state-input data magnitude (“Relative Scale”) increases $\alpha$ and reduces $\|Z_{12}\|$ and $\sum\gamma^{(i)}$, so the ratio $\|Z_{12}\|/\alpha$ decreases monotonically---as predicted by \eqref{eq:global}.}
  \label{fig:cz-scale}
\end{figure}

\subsection{Explicit Bounds For Constrained Matrix Zonotopes}
When the model set is a CMZ, the equality constraints restrict the feasible coefficients and can yield tighter bounds. For a CMZ $\mathcal{N} = \bigl\langle C_{\mathcal{N}}, \tilde{G}_{\mathcal{N}},\hat{A}_{\mathcal{N}}, \hat{b}_{\mathcal{N}} \bigr\rangle$, the bounds are obtained by finding $\xi_{max}$ such that
\begin{equation}
    \xi_{max} = \max_{\boldsymbol{\xi}} f(\xi) \quad\text{subject to} \quad \hat{A}_{\mathcal{N}}\xi = \hat{b}_{\mathcal{N}} \ \land \ \| \xi \|_{\infty} \le 1 .
    \label{eq:max-xi}
\end{equation}
where $f(\xi)$ is defined in~\eqref{eq:fxi}. One approach to finding $\xi_{max}$ is to enumerate all extreme points of the constraint polytope~\cite{Schneider_2013,ziegler_lecs}. Since the number of extreme points can be prohibitively large~\cite{patrix_cube_cuts}, we instead employ a more efficient linear programming approach. We transform~\eqref{eq:max-xi} into a linear fractional program, which is then converted into an equivalent linear program via the Charnes--Cooper transformation~\cite{charnes-cooper,Boyd_Vandenberghe_2004}. The procedure is outlined below.

From ~\eqref{eq:fxi} we can re-write ~\eqref{eq:max-xi} as follows:
\begin{multline}
    \xi_{max} = \max_{\boldsymbol{\xi}} \frac{\mu^\top |\xi|}{\sigma_{\text{min}}(C) -\gamma^\top |\xi|} \quad\text{subject to} \\ \hat{A}_{\mathcal{N}}\xi = \hat{b}_{\mathcal{N}} \ \land \ \| \xi \|_{\infty} \le 1 \ .
    \label{eq:cmz-bound-1}
\end{multline}
Assuming $\sigma_{\text{min}}(C) \ge \gamma^\top |\xi|$, ~\eqref{eq:cmz-bound-1} is equivalent to:
\begin{multline}
    \xi_{max} = \min_{\boldsymbol{\xi}} \frac{\sigma_{\text{min}}(C) -\gamma^\top |\xi|}{\mu^\top |\xi|} \quad\text{subject to} \\ \hat{A}_{\mathcal{N}}\xi = \hat{b}_{\mathcal{N}} \ \land \ \| \xi \|_{\infty} \le 1 \ \land \ \gamma^\top |\xi| \le \sigma_{\text{min}}(C).
    \label{eq:cmz-bound-2}
\end{multline}
Since $\xi$ can be expressed as $\xi_+ - \xi_-$ where $\xi_+ \ge \mathbf{0}$ and $\xi_- \ge \mathbf{0}$, therefore ~\eqref{eq:cmz-bound-2} is equivalent to:
\begin{align}
    \xi_{max} = \min_{\boldsymbol{\xi}} & \notag \frac{\sigma_{\text{min}}(C) -\gamma^\top (\xi_+ + \xi_-)}{\mu^\top (\xi_+ + \xi_-)} \quad\text{subject to} \\  & \notag \gamma^\top \xi_+ + \gamma^\top \xi_- \le \sigma_{\text{min}}(C) \ \land \hat{A}_{\mathcal{N}}(\xi_+ - \xi_-) = \hat{b}_{\mathcal{N}} \ \land \\ & \| \xi_+ - \xi_- \|_{\infty} \le \mathbf{1} \  \land \ \xi_+ \ge \mathbf{0} \ \land \ \xi_- \ge \mathbf{0} .
    \label{eq:cmz-bound-3}
\end{align}
Let $\bar{\xi} = 
\begin{bmatrix}
    \xi^T_+  \quad \xi^T_-
\end{bmatrix}^T$, then ~\eqref{eq:cmz-bound-3} is equivalent to:
\begin{align}
    \bar{\xi} = \min_{\boldsymbol{\xi}} & \notag \frac{\sigma_{\text{min}}(C) + \begin{bmatrix}
        -\gamma^\top & -\gamma^\top
    \end{bmatrix} \xi}{\begin{bmatrix}
        \mu^\top & \mu^\top
    \end{bmatrix} \xi} \quad\text{subject to} \\ 
    & \notag \begin{bmatrix}
        \gamma^\top & \gamma^\top
    \end{bmatrix} \xi \le \sigma_{\text{min}}(C) \ \land
    \xi \ge \mathbf{0} \ \land \\ &\begin{bmatrix} \hat{A}_{\mathcal{N}} & -\hat{A}_{\mathcal{N}} \end{bmatrix} \xi = \hat{b}_{\mathcal{N}} \ \land \
    \begin{bmatrix}
        I_{\gamma_{\mathcal{N}}} & -I_{\gamma_{\mathcal{N}}} \\
        -I_{\gamma_{\mathcal{N}}} & I_{\gamma_{\mathcal{N}}}
    \end{bmatrix} \xi \le \mathbf{1} .
    \label{eq:cmz-bound-4}
\end{align}

Using the Charnes--Cooper transformation~\cite{charnes-cooper,Boyd_Vandenberghe_2004}, \eqref{eq:cmz-bound-4} is converted into an equivalent linear program. After solving it, we recover $\bar{\xi}$, decompose it into the two halves $\xi_+$ and $\xi_-$, and set $\xi_{max} = \xi_+ - \xi_-$. The worst-case bound then becomes:
\begin{equation}
    \label{eq:cmz-worst-bounds}
    \max_{\boldsymbol{\xi}} \|\sin\Theta\| \leq \frac{\mu^\top |\xi_{max}|}{\sigma_{\text{min}}(C) -\gamma^\top |\xi_{\text{max}}|}.
\end{equation}

Before performing these calculations, one must verify that no feasible $\xi$ satisfies $\sigma_{\text{min}}(C) -\gamma^\top |\xi| \le 0$. This check can be carried out by solving the following linear program:
\begin{equation*}
    \xi_{test} = \min_{\boldsymbol{\xi}} -\gamma^\top |\xi| \quad \text{subject to} \quad \hat{A}\xi = \hat{b}_{\mathcal{N}} \ \land \ \| \xi \|_\infty \le 1
\end{equation*}
and checking if $\sigma_{\text{min}}(C) - \gamma^\top |\xi_{test}| \le 0$.

\section{Nullspace Matrix Zonotope: Tighter Explicit Bounds via Fewer Generators}
\label{sec:nmz}
Having derived an explicit Cai--Zhang bound for the CMZ, \eqref{eq:cmz-worst-bounds} reveals that both the numerator and the denominator accumulate contributions from \emph{all} $\gamma_{\mathcal{N}}$ generators. Two consequences follow. First, the bound is highly sensitive to the \emph{number of generators}: as $\gamma_{\mathcal{N}}$ grows, the accumulated terms become large. Second, the bound is attained at coefficient-space vertices (extreme combinations), which may introduce redundancy relative to the true worst case when many generators are present.

These observations motivate seeking an alternative representation that \emph{reduces the number of effective generators entering the explicit bound} while preserving set inclusion. To this end, we introduce the \emph{nullspace matrix zonotope (NMZ)}, obtained by over-approximating the feasible coefficient set of the CMZ by a zonotope of controllable complexity~\cite{raghuraman2022set}. Although the resulting matrix set is an over-approximation of the original CMZ, it typically has far fewer generators; we prove this at the end of this section.

\subsection{Nullspace Matrix Zonotope}
We begin with the definition of the nullspace matrix zonotope.

\begin{definition}[Nullspace Matrix Zonotope (NMZ)]
\label{def:nmz}
Given a CMZ $\mathcal{N}$ and a coefficient zonotope $\mathcal{Z}_\xi$ satisfying   $\Xi\subseteq\mathcal{Z}_\xi$~\eqref{eq:coefficent_space}, where $\mathcal{Z}_{\xi}$ is constructed from the nullspace of $\hat{A}_{\mathcal{N}}$ as detailed in Section ~\ref{sec:poly-to-zono-method}. The \emph{Nullspace Matrix Zonotope (NMZ)} associated with $(\mathcal{N},\mathcal{Z}_\xi)$ is: \begin{align}
\mathcal{M} &= \left\{ M \biggm| M = C_{\mathcal{N}} + \sum_{i=1}^p \xi^{(i)} G^{(i)}_{\mathcal{N}}, \ \ \xi \in \mathcal{Z}_\xi \right\} \\
&= \left\{ M \biggm| M = C_{\mathcal{N}} + \sum_{i=1}^p \left( c_{\xi}^{(i)} + (G_\xi \eta)^{(i)} \right) G^{(i)}_{\mathcal{N}}, \ \left\| \eta \right\|_{\infty} \le 1 \right\}
\label{eq:sub-with-zono} .
\end{align}
We regroup the terms into a new center $C_\mathcal{M}$ and new generator:
\begin{align}\label{new_nmz_center}
\mathcal{M} = \left\{ M \biggm| M = \left( C_{\mathcal{N}} + \sum_{i=1}^p c_{\xi}^{(i)} G^{(i)}_{\mathcal{N}} \right) + \sum_{i=1}^p (G_\xi \eta)^{(i)} G^{(i)}_{\mathcal{N}}, \left\| \eta \right\|_{\infty} \le 1 \right\} .
\end{align}
The new center is $C_\mathcal{M} = C_{\mathcal{N}} + \sum_{i=1}^p c_{\xi}^{(i)} G^{(i)}_{\mathcal{N}}$.
By expanding $(G_\xi \eta)^{(i)} = \sum_{j=1}^\gamma (G_\xi){(i,j)} \eta^{(j)}$ and swapping the summations, the new generator term becomes:
\begin{align}
\label{eq:new-mz-gens}
\sum_{i=1}^p \left( \sum_{j=1}^\gamma (G_\xi)_{(i,j)} \eta^{(j)} \right) G^{(i)}_{\mathcal{N}} = \sum_{j=1}^\gamma \eta^{(j)} \left( \sum_{i=1}^p (G_\xi)_{(i,j)} G^{(i)}_{\mathcal{N}} \right) .
\end{align}
This gives $\gamma$ new generators $G^{(j)}_{\mathcal{M}} = \sum_{i=1}^p (G_\xi)_{(i,j)} G^{(i)}_{\mathcal{N}}$.
Thus, the over-approximating matrix zonotope $\mathcal{M} = \left\langle {C}_{\mathcal{M}}, \ \{G^{(j)}_{\mathcal{M}}\}_{j=1}^\gamma \right\rangle$ is:
\begin{align}\label{eq:new_mz}
\mathcal{M} &= \left\{ M \biggm| M =C_\mathcal{M} + \sum_{j=1}^{\gamma} \eta^{(j)} G^{(j)}_\mathcal{M}, \ \eta^{(j)} \in [-1,1]\right\}.
\end{align}
Because $\Xi\subseteq\mathcal{Z}_\xi$ and the map $\xi\mapsto C_{\mathcal{N}}+\sum_i\xi^{(i)}G_{\mathcal{N}}^{(i)}$ is affine, it follows that $\mathcal{N}\subseteq\mathcal{M}$. The number of NMZ generators equals the number of columns $\gamma$ of $G_\xi$ and can be chosen significantly smaller than $p$.
The method for constructing $\mathcal{Z}_\xi$, described in Section~\ref{sec:poly-to-zono-method}, is closely related to the nullspace properties of $\hat{A}_\mathcal{N}$, hence the name.
\end{definition}

We now explain why the NMZ worst-case sin $\Theta$ distance \emph{upper-bounds} the CMZ worst case. Firstly, we refactor CMZ around the same center in \eqref{new_nmz_center}.
For any feasible CMZ coefficient $\xi\in\Xi$,
\begin{align}
\label{eq:cmz-refactor}
C_{\mathcal N}+\sum_{i=1}^p \xi^{(i)}G_{\mathcal N}^{(i)}
&= \Bigl(C_{\mathcal N}+\sum_{i=1}^p c_{\xi}^{(i)}G_{\mathcal N}^{(i)}\Bigr)
   + \sum_{i=1}^p\!\bigl(\xi^{(i)}-c_{\xi}^{(i)}\bigr)G_{\mathcal N}^{(i)} \nonumber\\
&= C_{\mathcal M}\ +\ \underbrace{\sum_{i=1}^p\!\bigl(\xi^{(i)}-c_{\xi}^{(i)}\bigr)G_{\mathcal N}^{(i)}}_{Z_{\mathrm{CMZ}}(\xi)} .
\end{align}
Thus the CMZ residual set at the common center is
\begin{equation}
\label{eq:Delta-cmz}
\Delta_{\mathrm{CMZ}}\ =\ \Bigl\{\,Z_{\mathrm{CMZ}}(\xi)=\sum_{i=1}^p\bigl(\xi^{(i)}-c_{\xi}^{(i)}\bigr)G_{\mathcal N}^{(i)}\ \Bigm|\ \xi\in\Xi\,\Bigr\}.
\end{equation}
Secondly, from \eqref{new_nmz_center}–\eqref{eq:new-mz-gens}, the corresponding residual set of NMZ is
\[
\Delta_{\mathrm{NMZ}}\ =\ \Bigl\{\,Z_{\mathrm{NMZ}}(\eta)=\sum_{i=1}^p (G_\xi\eta)^{(i)}G_{\mathcal N}^{(i)}\ \Bigm|\ \|\eta\|_\infty\le 1\,\Bigr\}.
\]
Third, take any $\xi\in\Xi$.
Since $\Xi\subseteq\mathcal{Z}_\xi$, there exists $\eta$ with $\|\eta\|_\infty\le 1$ such that
\begin{equation}
\label{eq:coeff-match}
\xi\ =\ c_\xi+G_\xi\eta \quad\Longrightarrow\quad \xi-c_\xi\ =\ G_\xi\eta .
\end{equation}
Substituting \eqref{eq:coeff-match} into \eqref{eq:cmz-refactor} gives
\[
Z_{\mathrm{CMZ}}(\xi)
= \sum_{i=1}^p (G_\xi\eta)^{(i)}G_{\mathcal N}^{(i)}
= Z_{\mathrm{NMZ}}(\eta)\ \in\ \Delta_{\mathrm{NMZ}},
\]
hence
\begin{equation}
\label{eq:Delta-inclusion}
\Delta_{\mathrm{CMZ}}\ \subseteq\ \Delta_{\mathrm{NMZ}} .
\end{equation}

Finally, we evaluate the bound in \eqref{eq:tight-right} at the \emph{fixed} center $C_{\mathcal M}$ for both sets. By \eqref{eq:Delta-inclusion} and monotonicity of the supremum with respect to set inclusion,
\begin{equation}
\label{eq:bound-ordering}
\sup_{Z\in \Delta_{\mathrm{CMZ}}}\ \text{RHS of \eqref{eq:tight-right}}
\ \le\
\sup_{Z\in \Delta_{\mathrm{NMZ}}}\ \text{RHS of \eqref{eq:tight-right}} .
\end{equation}
Therefore the NMZ worst-case sin $\Theta$ distance \emph{upper-bounds} the CMZ worst case under the common-center refactorization.
In particular, it is sufficient (and often simpler) to analyze the NMZ Cai–Zhang bound.

\begin{figure*}[!t]
  \centering
  \includegraphics[width=0.95\linewidth,page=1]{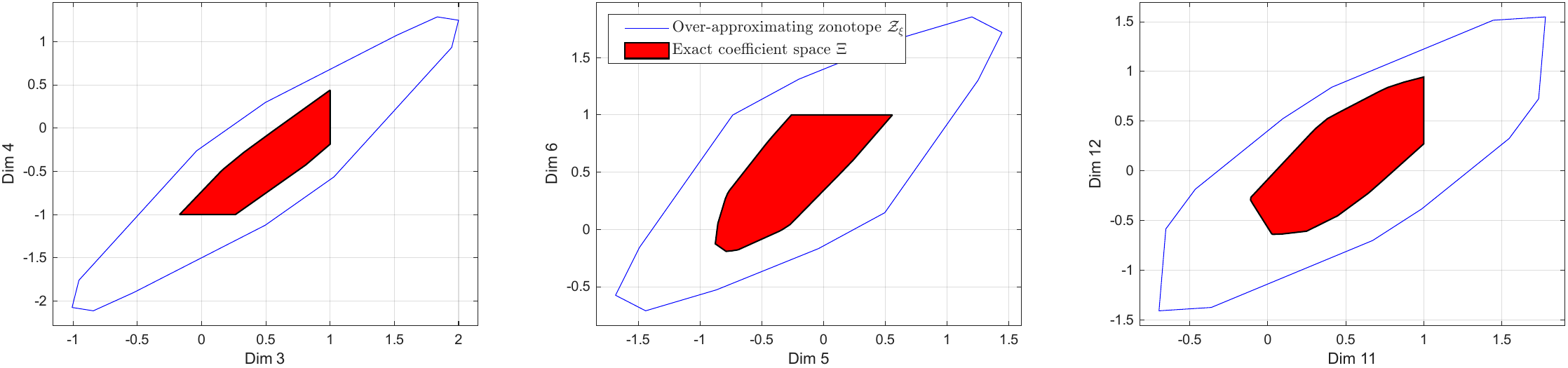}
  \caption{A comparison between the exact set of possible coefficients $\Xi$ and the over-approximating zonotope $\mathcal{Z}_\xi$, with the number of trajectories = 10, and the number of steps per trajectory = 3.}
  \label{fig:coeff-space-plots}
\end{figure*}

\subsection{Over-approximating the Coefficient Space}
\label{sec:poly-to-zono-method}
The NMZ definition above requires a coefficient zonotope $\mathcal{Z}_\xi$ that encloses the feasible set $\Xi$. We now describe how to construct such a zonotope efficiently using the nullspace of the constraint matrix.

Consider a CMZ with $\gamma$ generators. By Definition~\ref{def:new-cmz}, the generator coefficients satisfy $\hat{A}\xi = \hat{b}$ and $\left\| \xi \right\|_{\infty} \le 1$, which defines an H-representation of a convex polytope $\mathcal{P}$~\cite{girard2008zonotope,Schneider_2013, ziegler_lecs}.

A standard way to over-approximate $\mathcal{P}$ with a zonotope $\mathcal{Z_P}$ is to represent $\mathcal{P}$ as a constrained zonotope~\cite{scott2016constrained} and then apply the reduction method of~\cite{scott2016constrained} with the number of constraints set to zero.

A more refined approach, which yields better results in our experiments, involves projecting $\mathcal{P}$ onto the nullspace of $\hat{A}$ to obtain $\mathcal{P}'$, over-approximating $\mathcal{P}'$ with an interval or a zonotope, and then lifting the result back to the original coordinates.

Specifically, we compute the H-representation of $\mathcal{P}'$, over-approximate $\mathcal{P}'$ by an interval $\mathcal{I_P}$ via a sequence of linear programs, and then lift $\mathcal{I_P}$ back to the original coordinates to obtain $\mathcal{Z_P}$. The following theorem provides an efficient way to compute the H-representation of $\mathcal{P}'$.

\begin{theorem}
    \label{thm:lower-H-rep}
    $\mathcal{P}' = \{x \in \mathbb{R}^{\mathrm{nullity}(\hat{A})} \ | \ Qx \le s \}$, where $Q = \begin{bmatrix} I_{\gamma} \\ -I_{\gamma} \end{bmatrix}\hat{A}^\perp$, and $s = 1_{2\gamma} - \begin{bmatrix} I_{\gamma} \\ -I_{\gamma} \end{bmatrix}\xi_p$, where $\xi_p$ is a particular solution for the linear system $\hat{A}\xi = \hat{b}$.
\end{theorem}

\begin{proof}
    $\hat{A}\xi = \hat{b} \iff \xi = \xi_{\mathcal{N}} + \xi_p$ where $\xi_{\mathcal{N}} \in \mathcal{N}(\hat{A})$, and $\xi_p$ is a particular solution.
    $\xi_{\mathcal{N}} \in \mathcal{N}(\hat{A})$ is equivalent to $\xi_{\mathcal{N}}$ being a linear combination of a basis of $\mathcal{N}(\hat{A})$. Meaning $\xi_{\mathcal{N}}$ can \textbf{always} be written as $\hat{A}^{\perp}x_{\mathcal{N}}$ for some vector $x_{\mathcal{N}} \in \mathbb{R}^{\mathrm{nullity}(\hat{A})}$. Since the hypercube is represented by the inequality $$\begin{bmatrix} I_{\gamma} \\ -I_{\gamma} \end{bmatrix}\xi \le \mathbf{1}_{2\gamma} \iff \begin{bmatrix} I_{\gamma} \\ -I_{\gamma} \end{bmatrix}(\xi_{\mathcal{N}} + \xi_p) \le \mathbf{1}_{2\gamma} \ ,$$ then $x_{\mathcal{N}}$ exactly satisfies the inequality $$\begin{bmatrix} I_{\gamma} \\ -I_{\gamma} \end{bmatrix}(\hat{A}^{\perp}x_{\mathcal{N}} + \xi_p) \le \mathbf{1}_{2\gamma} \iff \begin{bmatrix} I_{\gamma} \\ -I_{\gamma} \end{bmatrix}\hat{A}^{\perp}x_{\mathcal{N}} \le \mathbf{1}_{2\gamma} - \begin{bmatrix} I_{\gamma} \\ -I_{\gamma} \end{bmatrix}\xi_p \ .$$
\end{proof}

Given the H-representation of $\mathcal{P}'$, it can be over-approximated by an interval $\mathcal{I_P}$~\cite{althoff2010reachability} by solving a sequence of linear programs subject to $Qx \le s$. The following theorem shows how to construct $\mathcal{Z_P}$ from $\mathcal{I_P}$.

\begin{theorem}
    \label{thm:over-appox-zono}
    $\mathcal{Z_P} = \underbrace{\xi_p + \hat{A}^\perp c_{\mathcal{I_P}}}_{\text{new center}} + \underbrace{\hat{A}^\perp \text{diag}(r_\mathcal{I_P})}_{\text{new generators}} \beta$, where $\|\beta\|_\infty \le 1$, and $c_{\mathcal{I_P}}$ and $r_{\mathcal{I_P}}$ are the center vector and the radius vector of $\mathcal{I_P}$, respectively, as in ~\cite{althoff2010reachability}.
\end{theorem}

\begin{proof}
    $\mathcal{P}' \subseteq \mathcal{I_P} \implies x_{\mathcal{N}} \in \mathcal{I_P}$. Since $\xi = \hat{A}^{\perp}x_{\mathcal{N}} + \xi_p$, therefore $\xi \in \hat{A}^{\perp}\mathcal{I_P} + \xi_p$. From ~\cite{althoff2010reachability} we can re-represent $\mathcal{I_P}$ as a zonotope $c_{\mathcal{I_P}} + \text{diag}(r_\mathcal{I_P})\beta$ where $\|\beta\|_\infty \le 1$, substituting gives us 
    \begin{multline}
        \xi = \hat{A}^{\perp}(c_{\mathcal{I_P}} + \text{diag}(r_\mathcal{I_P})\beta) + \xi_p \\ \iff \xi = \xi_p + \hat{A}^{\perp}c_{\mathcal{I_P}} + \hat{A}^{\perp}\text{diag}(r_\mathcal{I_P})\beta
    \end{multline}
    where $\|\beta\|_\infty \le 1$.
\end{proof}

\subsection{Application to Reachability Analysis}
We now apply the coefficient-space over-approximation to the data-driven model set $\mathcal{N}_\Sigma$.

Let $D = \begin{bmatrix} X_- \\ U_- \end{bmatrix}$ where $X_-$ and $U_-$ are as defined in ~\ref{def:input-state_trajectories}. Then by vectorizing each term in ~\eqref{eq:const-eq} we get
\begin{multline}
    \begin{bmatrix}
        \mathrm{vec}(G^{(1)}_{\mathcal{N}_w} D^\perp) &
        \dots &
        \mathrm{vec}(G^{(\gamma_{\mathcal{Z}_w}T)}_{\mathcal{N}_w} D^\perp)
    \end{bmatrix} \xi \\ = \mathrm{vec}((X_+ - C_{\mathcal{N}_w})D^\perp) .
\end{multline}
By definition ~\ref{def:new-cmz}, $\hat{A}_{\mathcal{N}_w} =     \begin{bmatrix}
        \mathrm{vec}(G^{(1)}_{\mathcal{N}_w} D^\perp) &
        \dots &
        \mathrm{vec}(G^{(\gamma_{\mathcal{Z}_w}T)}_{\mathcal{N}_w} D^\perp)
\end{bmatrix}$,
and $\hat{b}_{\mathcal{N}_w} = \mathrm{vec}((X_+ - C_{\mathcal{N}_w})D^\perp)$. \\
From~\eqref{eq:N-sigma} and~\eqref{eq:RN-cmz}, it follows that the constraints of $\mathcal{N}_\Sigma$ coincide with those of $\mathcal{N}_w$.

\begin{algorithm}[h]
    \caption{Nullspace-Reachability}
    \label{alg:new-DDRA-algo}  
    \begin{algorithmic}[1]
      \Require $\mathcal{N}_\Sigma = \left\langle C_{\mathcal{N}_\Sigma}, \{G_{\mathcal{N}_\Sigma}\}_{i = 0}^{\gamma_{\mathcal{Z}_w}T}, \hat{A}_{\mathcal{N}_\Sigma}, \hat{b}_{\mathcal{N}_\Sigma} \right\rangle$, initial set $\mathcal{X}_{0}$, process noise zonotope $\mathcal{Z}_w$, input zonotope $\mathcal{U}_k$, $\forall k = 0, \dots,N-1$
    \Ensure reachable sets $\mathcal{R}_{k}, \forall k = 1, \dots,N$
        \State $Q = \begin{bmatrix} I_{\gamma_{\mathcal{N}_\Sigma}} \\ -I_{\gamma_{\mathcal{N}_\Sigma}} \end{bmatrix} \hat{A}_{\mathcal{N}_\Sigma}^{\perp}$
        \State $\xi_p = \hat{A}_{\mathcal{N}_\Sigma}^{\dagger} \hat{b}_{\mathcal{N}_\Sigma}$ 
        \State $s = 1_{2\gamma_{\mathcal{N}_\Sigma}} - \begin{bmatrix} I_{\gamma_{\mathcal{N}_\Sigma}} \\ -I_{\gamma_{\mathcal{N}_\Sigma}} \end{bmatrix} \xi_p$
        \State $\mathcal{P}' = \left\{ x \Biggm| Q x \leq s \right\}$ \Comment{Theorem~\ref{thm:lower-H-rep}}
        \State $\mathcal{I_P} = \left(\mathcal{\underline{I}_P}, \mathcal{\overline{I}_P} \right) = \mathrm{Interval}(\mathcal{P}')$
        \State $c_{\mathcal{I_P}} = \frac{\mathcal{I_P}_L + \mathcal{I_P}_U}{2}$
        \State $r_{\mathcal{I_P}} = \frac{\mathcal{I_P}_U - \mathcal{I_P}_L }{2}$
        \State $\mathcal{Z_\xi} = \left\langle c_\xi, G_\xi \right\rangle = \left\langle \hat{A}_{\mathcal{N}_\Sigma}^\perp c_{\mathcal{I_P}} + \xi_p, \ \hat{A}_{\mathcal{N}_\Sigma}^\perp \mathrm{diag}(r_{\mathcal{I_P}}) \right\rangle$  \Comment{Theorem~\ref{thm:over-appox-zono}}
        \State $C_{\mathcal{P}_\Sigma} = C_{\mathcal{N}_\Sigma} + \sum_{i=1}^{\gamma_{\mathcal{Z}_w}T} c_{\xi}^{(i)} G^{(i)}_{\mathcal{N}_\Sigma}$
        \State $G_{\mathcal{P}_\Sigma}^{(j)} = \sum_{i=1}^{\gamma_{\mathcal{Z}_w}T}{(G_\xi)_{(i, j)}G_{\mathcal{N}_\Sigma}^{(i)}}, \quad \forall j \in \{1, \dots, \gamma_\xi \}$  \Comment{Apply ~\eqref{eq:new-mz-gens}}
        \State $\mathcal{P}_\Sigma = \left\langle C_{\mathcal{P}_\Sigma}, \{G_{\mathcal{P}_\Sigma}^{(i)}\}_{i = 0}^{\gamma_\xi} \right\rangle$  \Comment{Apply ~\eqref{eq:new_mz}}
        \State $\hat{\mathcal{R}}_{0} =\mathcal{X}_{0}$
        \For{$k = 0:N-1$} 
            \State $\hat{\mathcal{R}}_{k+1} =\mathcal{P}_{\Sigma} (\hat{\mathcal{R}}_{k} \times \mathcal{U}_{k}  ) \oplus  \mathcal{Z}_w$  \Comment{Apply ~\eqref{eq:reach-propagation}}
        \EndFor
    \end{algorithmic}
\end{algorithm}

Figure~\ref{fig:coeff-space-plots} illustrates how well a zonotope over-approximates the feasible coefficient set $\Xi$. In our experiments, increasing the amount of data improves the results because the nullity of $\hat{A}_{\mathcal{N}_w}$ becomes smaller relative to the dimension of $\xi$, reducing the dimension of $\mathcal{I_P}$. The following results quantify this effect.

\begin{lemma}
    \label{const-mat-rank}
    $\mathrm{rank}(\hat{A}_{\mathcal{N}_w}) = \mathrm{rank}(D^{\perp})\cdot \mathrm{rank}(G_{\mathcal{Z}_{w}})$
    where $\hat{A}_{\mathcal{N}_w} =     \begin{bmatrix}
        \mathrm{vec}(G^{(1)}_{\mathcal{N}_w} D^\perp) &
        \dots &
        \mathrm{vec}(G^{(\gamma_{\mathcal{Z}_w}T)}_{\mathcal{N}_w} D^\perp)
    \end{bmatrix}$, and $D$ is the data matrix, and $G_{\mathcal{Z}_{w}}$ is the generator matrix of the noise zonotope $\mathcal{Z}_w$.
    \label{thm:kron-rank}    
\end{lemma}

\begin{proof}
    From definition ~\ref{def:noise-matrix-zono}, the $k$-th column in $\hat{A}_{\mathcal{N}_w}$ is 
    \begin{align}
        \mathrm{vec}(G^{(k)}_{\mathcal{M}_w} D^\perp) &= \mathrm{vec}(
            \begin{bmatrix}
                0_{n \times (j-1)} & g_{\mathcal{Z}_w}^{(i)} & 0_{n \times (T-j)}
            \end{bmatrix} D^\perp
        ) \\
        &= \mathrm{vec}
        (\begin{bmatrix}
            D^{\perp}_{(j,1)}g_{\mathcal{Z}_{w}}^{(i)} & D^{\perp}_{(j,2)}g_{\mathcal{Z}_{w}}^{(i)} &
            \dots &
        \end{bmatrix}) \\ 
        &= \mathrm{vec}
        (\begin{bmatrix}
            D^{\perp^{\top}}_{(1,j)}g_{\mathcal{Z}_{w}}^{(i)} & D^{\perp^{\top}}_{(2,j)}g_{\mathcal{Z}_{w}}^{(i)} &
            \dots &
        \end{bmatrix})
        \label{eq:kron-cols}
    \end{align}
    where $k = j + (i-1)T$ such that $j$ is the remainder after dividing $k$ by $T$.

    The columns of $\hat{A}_{\mathcal{N}_w}$ as in ~\eqref{eq:kron-cols} can be re-ordered to give a matrix $\hat{A}_{\mathcal{N}_w}'$ such that:
    \begin{equation}
        \hat{A}_{\mathcal{N}_{w}(r:r+n, \ :)}' =
        \begin{bmatrix}
            D^{\perp^{\top}}_{(r,1)}g_{\mathcal{Z}_{w}}^{(1)} &
            \cdots &
            D^{\perp^{\top}}_{(r,1)}g_{\mathcal{Z}_{w}}^{(\gamma_{\mathcal{Z}_w})} &
            D^{\perp^{\top}}_{(r,2)}g_{\mathcal{Z}_{w}}^{(1)} & \cdots 
        \end{bmatrix}
        \label{eq:col-reorder}
    \end{equation}
    $\forall r \in \{1, \dots,\mathrm{nullity}(D)\}$. By definition $G_{\mathcal{Z}_w} = \begin{bmatrix} g_{\mathcal{Z}_w}^{(1)} \cdots g_{\mathcal{Z}_w}^{(\gamma_{\mathcal{Z}_w})} \end{bmatrix}$, therefore \eqref{eq:col-reorder} is equivalent to:
    \begin{equation}
        \hat{A}_{\mathcal{N}_{w}(r:r+n, \ :)}' =
        \begin{bmatrix}
            D^{\perp^{\top}}_{(r,1)} G_{\mathcal{Z}_w} & D^{\perp^{\top}}_{(r,2)} G_{\mathcal{Z}_w} &
            \cdots &
            D^{\perp^{\top}}_{(r,T)} G_{\mathcal{Z}_w}
        \end{bmatrix}
    \end{equation}
    and this is by definition equivalent to $\hat{A}_{\mathcal{N}_{w}}' = D^{\perp^{\top}} \otimes G_{\mathcal{Z}_w}$. Therefore ~\cite{matrix-cookbook} :
    $$\mathrm{rank}(\hat{A}_{\mathcal{N}_w}') = \mathrm{rank}(D^{\perp^{\top}}) \cdot \mathrm{rank}(G_{\mathcal{Z}_w}) = \mathrm{rank}(D^{\perp}) \cdot \mathrm{rank}(G_{\mathcal{Z}_w})$$
    Since changing the order of columns does not change the rank of a matrix, therefore: $$\mathrm{rank}(\hat{A}_{\mathcal{N}_w}) = \mathrm{rank}(\hat{A}_{\mathcal{N}_w}') = \mathrm{rank}(D^{\perp}) \cdot \mathrm{rank}(G_{\mathcal{Z}_w}) .$$
\end{proof}

Using Lemma~\ref{thm:kron-rank}, we can establish a relation between the nullity of $\hat{A}_{\mathcal{N}_w}$, the rank of $D$, and the rank of $G_{\mathcal{Z}_w}$.

\begin{corollary}
    $\mathrm{rank}(\hat{A}_{\mathcal{N}_w}) = \mathrm{nullity}(D) \cdot \mathrm{rank}(G_{\mathcal{Z}_w})$, where $\hat{A}_{\mathcal{N}_w} =     \begin{bmatrix}
        \mathrm{vec}(G^{(1)}_{\mathcal{N}_w} D^\perp) &
        \dots &
        \mathrm{vec}(G^{(\gamma_{\mathcal{Z}_w}T)}_{\mathcal{N}_w} D^\perp)
    \end{bmatrix}$, and $D$ is the data matrix and $G_{\mathcal{Z}_w}$ is the generator matrix of $\mathcal{Z}_w$.
\end{corollary}

\begin{theorem}
    $\mathrm{nullity}(\hat{A}_{\mathcal{N}_w}) = \gamma_{\mathcal{Z}_w}T - (T - \mathrm{rank}(D)).\mathrm{rank}(G_{\mathcal{Z}_w})$, where $\hat{A}_{\mathcal{N}_w} =     \begin{bmatrix}
        \mathrm{vec}(G^{(1)}_{\mathcal{N}_w} D^\perp) &
        \dots &
        \mathrm{vec}(G^{(\gamma_{\mathcal{Z}_w}T)}_{\mathcal{N}_w} D^\perp)
    \end{bmatrix}$, and $D$ is the data matrix and $G_{\mathcal{Z}_w}$ is the generator matrix of $\mathcal{Z}_w$.
\end{theorem}

\begin{proof}
    This directly follows from the rank-nullity theorem knowing that the number of columns of $\hat{A}_{\mathcal{N}_w}$ is $\gamma_{\mathcal{Z}_w}T$.
\end{proof}

\begin{remark}
    If $D$ has full row rank, then
    \begin{equation}
        \mathrm{nullity}(\hat{A}_{\mathcal{N}_w}) = \gamma_{\mathcal{Z}_w}T - (T - n - m).\mathrm{rank}(G_{\mathcal{Z}_w})
        \label{eq:nulllity-full-rank} .
    \end{equation}
\end{remark}

\begin{remark}
    If $\mathcal{Z}_w$ has only one generator (i.e. $G_{\mathcal{Z}_w}$ is a single column), then $\gamma_{\mathcal{Z}_w} = 1$ and $\mathrm{rank}(G_{\mathcal{Z}_w}) = 1$. This simplifies ~\eqref{eq:nulllity-full-rank} to be
    \begin{equation}
        \mathrm{nullity}(\hat{A}_{\mathcal{N}_w}) = T - (T - n - m) = n + m .
    \end{equation}
\end{remark}

An additional reason why increasing $T$ improves results is that the number of generators of $\mathcal{P}_\Sigma$ becomes small relative to the number of generators of $\mathcal{N}_\Sigma$. From~\eqref{eq:new-mz-gens} and Theorem~\ref{thm:over-appox-zono}, the number of generators of $\mathcal{P}_\Sigma$ equals $\mathrm{nullity}(\hat{A}_{\mathcal{P}_\Sigma})$, whereas the number of generators of $\mathcal{N}_\Sigma$ is $\gamma_{\mathcal{Z}_w}T$.

\begin{remark}[Extension to piecewise affine systems]
Although the above analysis focuses on LTI dynamics, the nullspace matrix zonotope construction and the resulting reachability procedure extend naturally to hybrid settings. In particular, the same NMZ method can be embedded in a piecewise affine (PWA) framework, as discussed in~\cite{xie2025data}. We illustrate this extension with an example in Section~\ref{sec:exp}.
\end{remark}


\subsection{Perturbation Bound for the Nullspace Matrix Zonotope}
Having constructed the NMZ, we now derive an explicit Cai--Zhang bound tailored to its structure.
Since~\eqref{eq:new_mz} is a standard matrix zonotope, the worst-case bound framework from~\eqref{eq:vertex} applies directly, adapted to the new center $C_\mathcal{M}$ and the new generators $G^{(j)}_\mathcal{M}$.

Let $U_\mathcal{M}$ and $V_{\mathcal{M},\perp}$ be the left singular basis and the orthogonal complement of the right singular basis of the new center $C_\mathcal{M}$, respectively. The perturbation $E(\boldsymbol{\eta})$ is now $E(\boldsymbol{\eta}) = \sum_{j=1}^{\gamma} \eta^{(j)} G^{(j)}_\mathcal{M}$.

We define the new component norms based on the $\gamma$ new generators:
\begin{align}
\gamma_{\mathcal{M}}^{(j)} &= \|G^{(j)}_\mathcal{M}\|_2 = \left\| \sum_{i=1}^p (G_\xi)_{(i,j)} G^{(i)}_{\mathcal{N}} \right\|_2 , \\
\mu_{\mathcal{M}}^{(j)} &= \|U_\mathcal{M}^{\top} G^{(j)}_\mathcal{M} V_{\mathcal{M},\perp}\|_2 = \left\| U_\mathcal{M}^{\top} \left( \sum_{i=1}^p (G_\xi)_{(i,j)} G^{(i)}_{\mathcal{N}} \right) V_{\mathcal{M},\perp} \right\|_2 .
\end{align}
The worst-case bound for the subspace perturbation of $\mathcal{M}$ is then given by:
\begin{equation}
\label{eq:vertex-M}
\max_{\boldsymbol{\eta}} \|\sin\Theta\| \leq \frac{\sum_{j=1}^{\gamma} \mu_{\mathcal{M}}^{(j)}}{\sigma_{\min}(C_\mathcal{M}) - \sum_{j=1}^{\gamma} \gamma_{\mathcal{M}}^{(j)}} .
\end{equation}
This bound is valid when $\sum_{j=1}^{\gamma} \gamma_{\mathcal{M}}^{(j)} < \sigma_{\min}(C_\mathcal{M})$.

\paragraph*{Bound under Data-Driven Structure}
If the original generators $G^{(i)}_{\mathcal{N}}$ possess the shared structure $G^{(i)}_{\mathcal{N}} = G_{\mathcal{M}_w}^{(i)}H$ as in \eqref{eq:zonotope-structure}, the new generators $G^{(j)}_\mathcal{M}$ inherit this structure:
$$
G^{(j)}_\mathcal{M} = \sum_{i=1}^p (G_\xi)_{(i,j)} G^{(i)}_{\mathcal{N}} = \left( \sum_{i=1}^p (G_\xi)_{(i,j)} G_{\mathcal{M}_w}^{(i)} \right) H \triangleq \tilde{G}^{(j)}_{\mathcal{M}_w} H .
$$
By applying the factorization logic from \eqref{eq:global}, we derive a refined bound based on the basis of $C_\mathcal{M}$:
\begin{align}
\label{eq:global-M}
\|\sin\Theta\| &\leq \frac{\kappa_\mathcal{M} \sum_{j=1}^{\gamma} \mu_{w,\mathcal{M}}^{(j)}}{\sigma_{\min}(C_\mathcal{M}) - \sum_{j=1}^{\gamma} \gamma_{\mathcal{M}}^{(j)}},\\
\text{where } \kappa_\mathcal{M} &= \|H V_{\mathcal{M},\perp}\|_2, \quad \mu_{w,\mathcal{M}}^{(j)} = \|U_\mathcal{M}^{\top} \tilde{G}^{(j)}_{\mathcal{M}_w}\|_2.
\end{align}

Compared with the CMZ bound \eqref{eq:cmz-worst-bounds}, \eqref{eq:global-M} accumulates only over the $\gamma$ \emph{effective} generators of the NMZ, rather than over all $\gamma_{\mathcal{N}}$ original generators. When the equality constraints compress the feasible coefficient set into a low-dimensional subspace, it is often possible to choose $\mathcal{Z}_\xi$ so that $\gamma\ll \gamma_{\mathcal{N}}$, thereby reducing both the numerator and the denominator accumulations and yielding a tighter bound. 

\begin{figure}[t]
  \centering
  \includegraphics[width=0.9\linewidth,page=1]{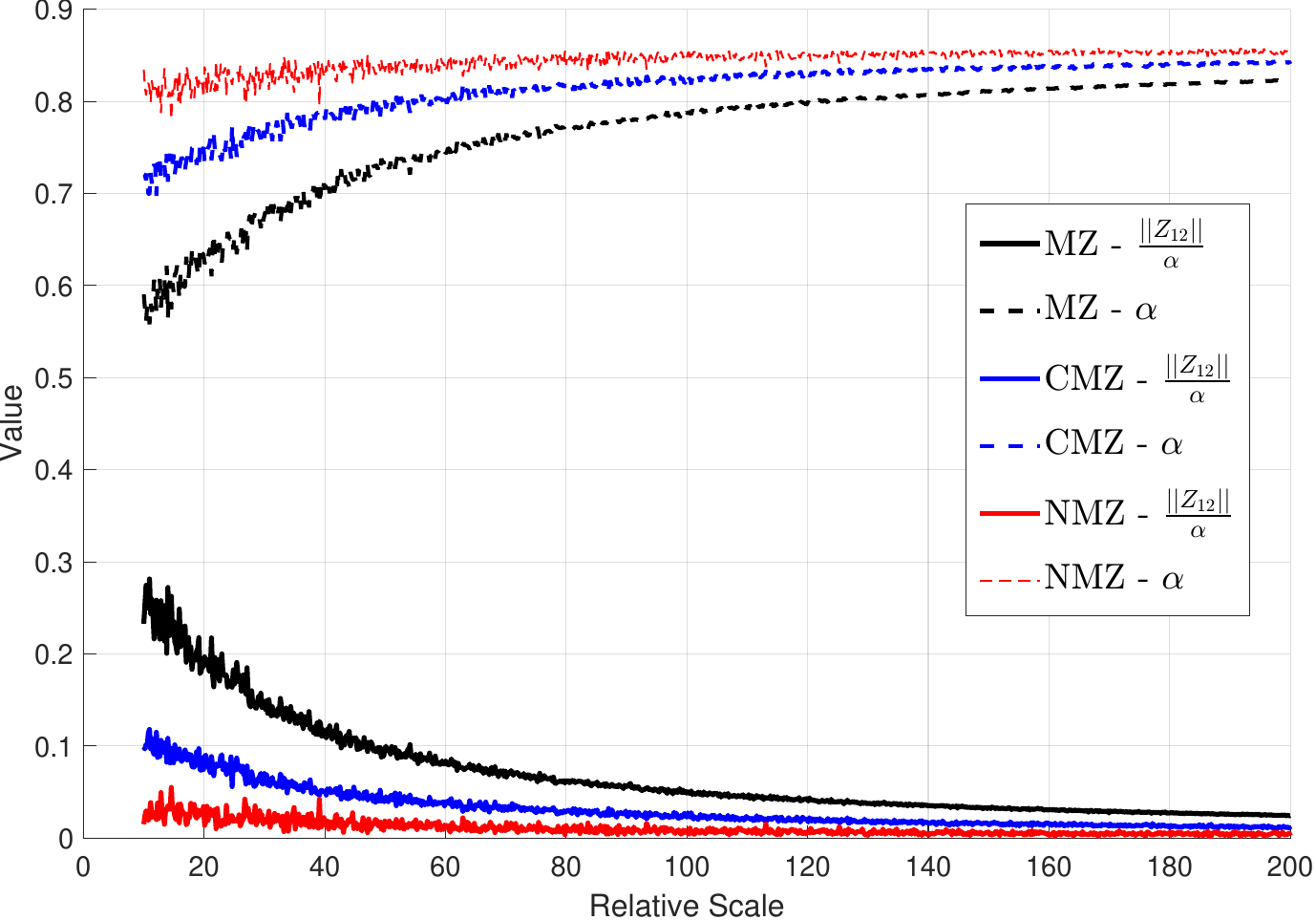}
  \caption{\textbf{CMZ, MZ, and NMZ under state-input data scaling.} Across scales, NMZ consistently attains the smallest ratio, reflecting the benefit of fewer effective generators.}
  \label{fig:three-methods}
\end{figure}
Figure~\ref{fig:three-methods} visualizes this effect in a controlled scaling experiment with fixed noise. The NMZ—by aggregating only over the $\gamma$ effective generators achieves a tighter curve than CMZ (and the plain MZ).

\section{Experimental Setup and Results}
\label{sec:exp}
We evaluate the proposed NMZ on a linear system and a piecewise affine hybrid system, comparing it against MZ- and CMZ-based propagation.
All experiments are implemented in MATLAB with the CORA toolbox~\cite{althoff2015cora} and executed on a Lenovo laptop equipped with an Intel 64-bit processor (1.5\,GHz) and 32\,GB of RAM, using a single CPU core without GPU acceleration.

\begin{figure*}[!t]
    \centering
    \includegraphics[width=0.95\linewidth]{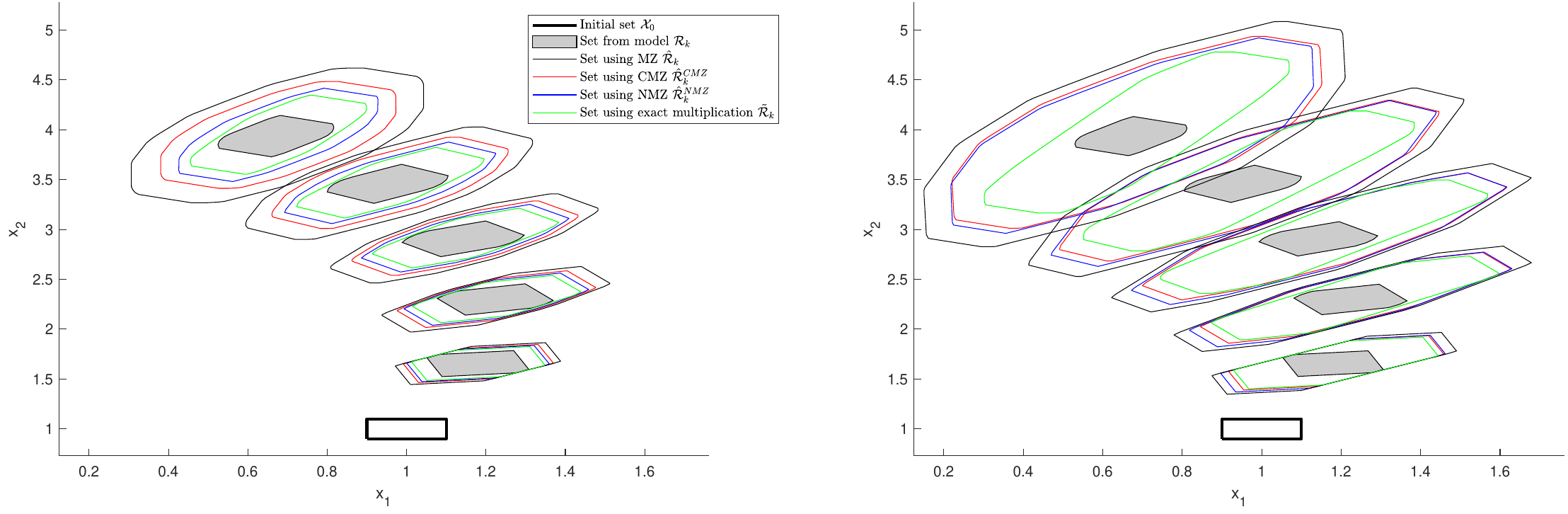}
    \caption{Comparison of reachable-set projections using matrix zonotope, constrained matrix zonotope, the nullspace matrix zonotope and exact multiplication. The left subfigure shows results with $T=50$ initial trajectories, while the right subfigure shows results with $T=30$ initial trajectories.}
    \label{fig:reachable-sets-comparison}
\end{figure*}

\begin{figure}[t]
  \centering
\includegraphics[width=\linewidth,page=1]{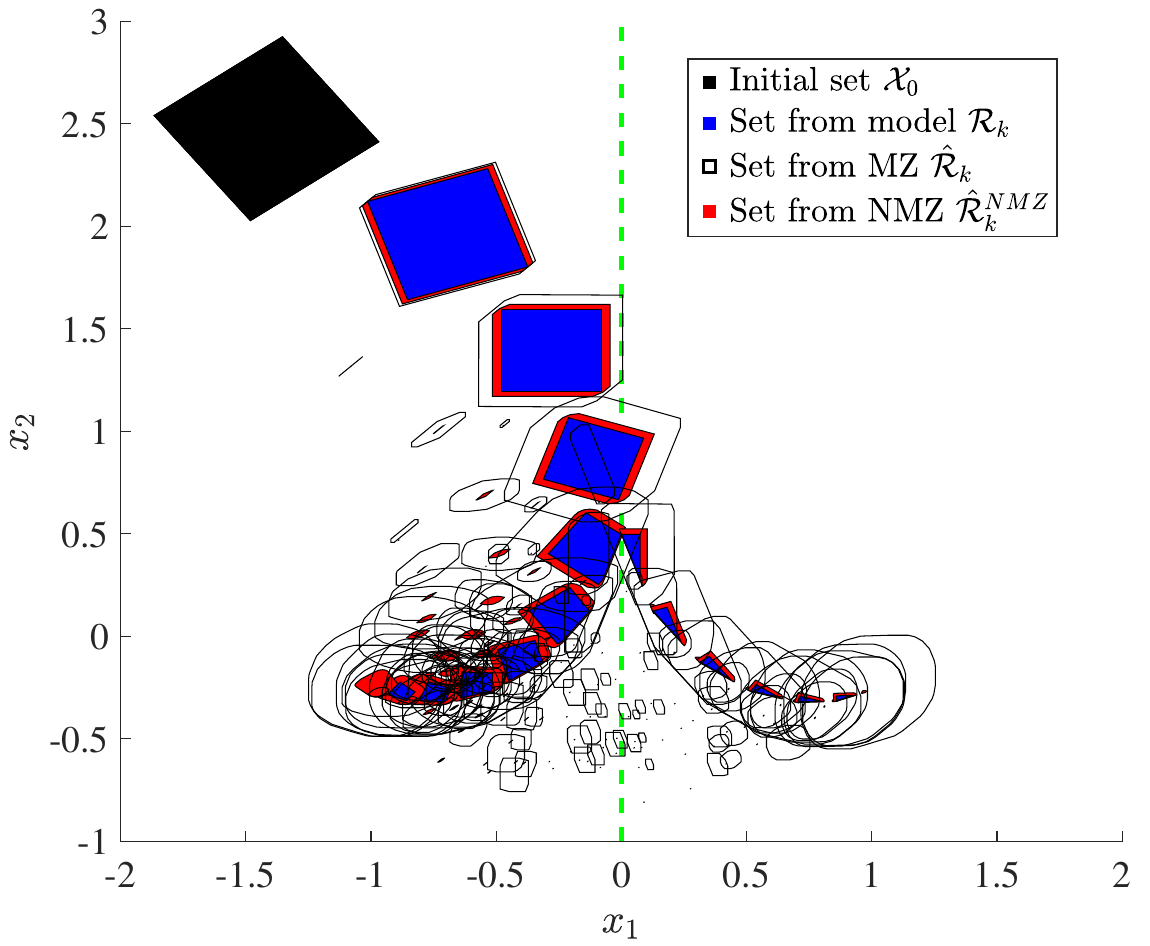}
\caption{\textbf{Piecewise affine system (switching at $x_1=0$).} Initial set $\mathcal{X}_0$ (black); model-based set $\mathcal{R}_k$ (blue); MZ-based set $\hat{\mathcal{R}}_k$ (black solid line); NMZ-based set $\hat{\mathcal{R}}^{\mathrm{NMZ}}_k$ (red). The NMZ yields a tighter outer approximation.}
  \label{fig:hybrid-nmz-mz}
\end{figure}

Consider a five-dimensional discrete-time linear system with initial state set $\mathcal{X}_0 = \langle \mathbf{1}_5, 0.1 I_5 \rangle$, a one-dimensional control input bounded by a zonotope
$\mathcal{U} = \langle 10, 0.25 \rangle$, and process noise $\mathcal{W} = \langle \mathbf{0}_5, \mathrm{diag}([1, 1.1, 1.3, 1, 1.5]) \rangle$. The system matrices are:

\begin{equation*}
A = \begin{bmatrix}
0.9323 & -0.189 & 0 & 0 & 0 \\
0.189 & 0.9323 & 0 & 0 & 0 \\
0 & 0 & 0.8596 & 0.04302 & 0 \\
0 & 0 & -0.04302 & 0.8596 & 0 \\
0 & 0 & 0 & 0 & 0.9048
\end{bmatrix}
\end{equation*}
\begin{equation*}
B^T = \begin{bmatrix}
0.04363 ,0.05327, 0.04754, 0.04528,0.04758
\end{bmatrix}
\end{equation*}

By varying the number of input--state trajectories in \eqref{def:input-state_trajectories}, we control the number of generators of the MZ and CMZ; the proposed NMZ is constructed from the same data via the
nullspace-based coefficient over-approximation described in Section~\ref{sec:nmz}.
We then compute reachable sets using three propagation modes:
(i) MZ-based propagation \(\hat{\mathcal{R}}_k\),
(ii) CMZ-based propagation \(\hat{\mathcal{R}}^{\mathrm{CMZ}}_k\), and
(iii) NMZ-based propagation \(\hat{\mathcal{R}}^{\mathrm{NMZ}}_k\).
As a reference, we include an exact set $\tilde{\mathcal{R}}_k$ obtained by
\(\text{CMZ}\times\) constrained polynomial zonotope (CPZ)~\cite{kochdumper2023constrained}, which yields a generally nonconvex exact
representation~\cite{zhang2025data}; see Figure~\ref{fig:reachable-sets-comparison}.

When the number of generators of the MZ and CMZ is small, we propagate with a generous reachable-set reduction
order~\cite{girard2005reachability,rego2025novel} of \(4000\). In the right subplot of Figure~\ref{fig:reachable-sets-comparison}, the reachable
set of the NMZ is slightly larger than that of the CMZ, as expected, because the NMZ over-approximates the CMZ. For step \(k{=}5\) (Table~\ref{tab:runtime}), the NMZ is \(\mathbf{7.17\times 10^4}\) times faster
than the CMZ. \emph{Crucially}, despite sharing the same unconstrained zonotope structure---and thus very
similar runtime---with the MZ (the MZ is only \(\approx\!1.9\times\) slower), the NMZ produces a \emph{much smaller} reachable set \(\hat{\mathcal{R}}^{\mathrm{NMZ}}_k\) than \(\hat{\mathcal{R}}_k\) for
comparable generator budgets. This represents a key advantage of the NMZ over the plain MZ.

As the number of generators of the MZ and CMZ grows, unreduced propagation becomes impractical. Under a uniform
reachable-set generator cap of \(1000\) (left subplot), the NMZ is \(\mathbf{4.16\times 10^4}\) times faster
than the CMZ, while remaining tighter than the CMZ at the same reduction order and noticeably smaller than the MZ.
The efficiency gap follows from the underlying structure: multiplication of the MZ or NMZ by a zonotope yields an unconstrained zonotope
with no additional constraints, whereas multiplication of the CMZ by a CPZ must carry equality constraints throughout, which
incurs substantial overhead. Overall, the NMZ preserves the simplicity and speed of the MZ
while delivering significantly less conservative reachable sets for a given generator budget.

We additionally evaluate the methods on a hybrid system with a switching surface at $x_1=0$ under strong noise~\cite{xie2025data}. As shown in Figure~\ref{fig:hybrid-nmz-mz}, the sets $\hat{\mathcal{R}}^{\mathrm{NMZ}}_k$ (red) produced by the NMZ yield a visibly tighter outer approximation than the sets $\hat{\mathcal{R}}_k$ (black solid line) produced by the MZ, while remaining close to the model-based set $\mathcal{R}_k$ (blue). Hybrid systems impose stricter demands on the tightness of the model set: errors near the switching boundary (at $x_1 = 0$ in Figure~\ref{fig:hybrid-nmz-mz}) can lead to exponentially accumulated propagation errors, which severely degrade the subsequent reachable-set computation~\cite{xie2025data}. As shown in Figure~\ref{fig:hybrid-nmz-mz}, although the true system is convergent, the reachable set computed with the MZ-based method becomes divergent, whereas the reachable set obtained with the NMZ-based method maintains convergent behavior.

\begin{table}[t]
\centering
\caption{Runtime (seconds) to compute \(\hat{\mathcal{R}}^{(\cdot)}_5\).}
\label{tab:runtime}
\begin{tabular}{lcccccc}
\hline
Scenario & order & MZ & CMZ & NMZ & Exact \\
\hline

 Figure~\ref{fig:reachable-sets-comparison} right & 4000 & 0.0065 & 145.73 & 0.0035 & 5.281  \\
 Figure~\ref{fig:reachable-sets-comparison} left  & 1000 & 0.0091 & 344.36 & 0.0048 & 95.571  \\
\hline
\end{tabular}
\end{table}

\section{Conclusion}\label{sec:conclus}

This paper introduced a matrix-perturbation viewpoint for data-driven reachability analysis. Leveraging spectral
perturbation theory, in particular the Cai--Zhang bounds, we quantified how noise-induced subspace rotations and spectral distortions propagate through the generators of a matrix zonotope. To alleviate the computational burden of the CMZ, we developed a CMZ-to-MZ
approximation pipeline that preserves set-containment guarantees while enabling efficient computation with markedly reduced conservatism relative to the plain MZ. We further proposed the \emph{nullspace matrix zonotope} (NMZ), which retains the unconstrained zonotope structure of the MZ yet delivers substantially tighter outer approximations than the CMZ when the number of generators of the reachable set is limited.

In future work, we plan to (i) exploit perturbation bounds to design \emph{maximal inner approximations} for data-driven reachability; (ii) extend the coefficient-space framework to hybrid systems, for example through a \emph{hybrid matrix zonotope} for reachability analysis of hybrid systems; and (iii) deploy the NMZ in practical settings, particularly real-time and time-varying systems, where its computational efficiency and reduced conservatism can improve data-driven verification and control.




\bibliographystyle{ACM-Reference-Format}
\bibliography{BibTex_2025}

\end{document}